\documentclass[runningheads]{llncs}
\usepackage[T1]{fontenc}
\usepackage{xcolor}
\usepackage{graphicx}
\usepackage{algorithm}
\usepackage{algpseudocode}
\usepackage{amsmath}
\begin{document}
\title{Engineering Resilience: An Energy-Based Approach to Sustainable Behavioural Interventions
}
\titlerunning{Engineering Resilience}
%
\author{Arpitha Srivathsa Malavalli\inst{1}\orcidID{0000-0002-5491-5535} \and
Karthik Sama\inst{1}\orcidID{0000-0002-5538-916X} \and
Janvi Chhabra\inst{1}\orcidID{0000-0002-4642-7726} \and
Pooja Bassin\inst{1}\orcidID{0000-0002-0611-8734}\and
Srinath Srinivasa\inst{1}\orcidID{0000-0001-9588-6550}}
\authorrunning{Malavalli et al.}
%
\institute{International Institute of Information Technology Bangalore, Bengaluru, Karnataka - 560100}
\maketitle              
\begin{abstract}
Addressing complex societal challenges—such as improving public health, fostering honesty in workplaces, or encouraging eco-friendly behaviour—requires effective nudges to influence human behaviour at scale. Intervention science seeks to design such nudges within complex societal systems. While interventions primarily aim to shift the system toward a desired state, less attention is given to the sustainability of that state, which we define in terms of resilience: the system's ability to retain the desired state even under perturbations. In this work, we offer a more holistic perspective to intervention design by incorporating a nature-inspired postulate—lower-energy states tend to exhibit greater resilience— as a regularization mechanism within intervention optimization to ensure that the resulting state is also sustainable. Using a simple agent-based simulation where commuters are nudged to choose eco-friendly options (e.g., cycles) over individually attractive but less eco-friendly ones (e.g., cars), we demonstrate how embedding lower energy postulate into intervention design induces resilience. The system energy is defined in terms of motivators that drive its agent’s behaviour. By inherently ensuring that agents are not pushed into actions that contradict their motivators, the energy-based approach helps design effective interventions that contribute to resilient behavioural states
\end{abstract}
\section{Introduction}
Nudges play a pivotal role in driving behavioural changes in numerous contexts-- be it guiding communities towards adopting energy-efficient solutions, encouraging consumers to choose specific products or motivating patients to prioritize their physical and mental well-being~\cite{BehChange_EnergyUse,BehChange_tryProcedures}. When populations are conceptualized as complex systems of rational agents making decisions, nudges serve as interventions that alter these decisions and steer the system towards desirable collective states—such as widespread adoption of eco-friendly commuting or greater patient compliance in healthcare.

Interventions are often costly, time-consuming, and difficult to reverse. Therefore, Intervention designers typically rely on simulations to evaluate multiple intervention designs before implementation~\cite{InterventionMAS,InterventionMAS2}. These interventions are designed with the primary aim of creating huge impact i.e large change towards the desired state. However, an equally important but often overlooked criterion is sustainability—the ability of the system to maintain the achieved state.

Sustainability is commonly understood in terms of longevity -- the maintenance of a desirable state over an extended period. However, when evaluating the impact of interventions through simulations, the time taken by the simulation to reach a final state offers little insight into the sustainability of that state. In the context of intervention science, the concept of \textit{resilience} offers a complementary perspective on sustainability. Resilience is defined as a system's capacity to absorb disturbances or perturbations while continuing to function and adapt to changes~\cite{walker2004resilience,ResinSus1}. If the final state is not resilient, small perturbations can undo the behavioural change achieved, and all the resources spent on implementing the intervention get wasted. Thus, it is important to aim for both, resilience and impact, when designing and evaluating interventions.

While resilience has been studied extensively, much of this work remains domain-specific—whether in telecommunications~\cite{ResilienceTelecommunication}, psychology~\cite{ResiliencePostTrauma}, or organizational behaviour~\cite{LeanSigSigmaResilience}. They also typically addresses resilience as the primary design objective rather than a property that must be achieved alongside the actual intervention objective .Our work draws inspiration more fundamental sustainability principle, observed in nature — systems tend to evolve toward "low-energy" configurations, which are highly sustainable~\cite{TheoryofBeing}. The attainment of hydrostatic equilibrium, electron emissions, and the formation of molecular structures are all examples of systems settling into low-energy states, supporting the postulate that such states are inherently more resilient and sustainable. Motivated by this principle, this work explores whether energy minimization can be incorporated into intervention design to achieve sustainable behavioural change..

The Fogg Behaviour Model~\cite{FoggBehaviourModel} lists key motivators that  shape human behaviour: pleasure/pain, hope/fear, and social acceptance/rejection (conformity). 
In behavioural systems, the dissonance experienced by agents when their actions conflict with their key motivators, can be considered an analogue to energy. We focus on conformity or an agent's need to align with majority as its dominant motivator. When their actions diverge from their neighbours' behaviours~\cite{solomon1956conformity,CognitiveDissconance1}, their energy increases with increased dissonance. The aggregate dissonance across agents represents the system's total energy. While nudging agents towards desired behaviours, interventions may force systems to high-energy states that are not resilient to perturbations.

To prevent this, we propose that interventions should not only aim to maximise impact but also embed constraints that prevent energy increase.  By designing interventions such that systems remain in low-energy configurations, while agents are steered towards desired behaviours, both impact and resilience can be achieved.

In this work, we investigate these ideas through two hypotheses:

\textbf{H1:} The lower the energy of a complex social system, the higher its resilience to targeted attacks and random perturbations.

\textbf{H2:} Incorporating energy-based constraints into intervention design yields more sustainable outcomes.

To test these hypotheses, we simulate a population of commuters who must choose between cycling and driving cars. While cars are the convenient , self serving choice, interventions have to be designed to nudge the population to choose the more eco-friendly option of cycles. 
Using agent-based models where conformity is the primary motivator, we show that configurations in which agents are at low energy - that is, experience little or no dissonance due to minimal differences in beliefs with their neighbours -  display high resilience to a range of perturbation scenarios. We also demonstrate how the design of the popular Connect People Intervention~\cite{ConnectPeopleIntervention} can be tweaked to include energy-aware constraints to achieve comparable impact but substantially better resilience.

Thus, this work proposes using system energy as a key metric to evaluate intervention designs for sustainability and resilience. It introduces the perspective of energy-aware intervention design, which maintains low energy while steering agents towards desired behaviours.

\section{Motivation}\label{sec:context}

\subsection{Resilience-based Interpretation of Sustainability}
The primary objective of a nudge/intervention is to push a population towards a certain desired behaviour. Additionally, the policy makers would also expect that this change persists, that the individuals consistently choose to make the desired choice i.e. the impact of interventions needs to be sustainable.
 
Many authors~\cite{colocousis2017long} agree that time forms one of the key dimensions of sustainability. Various measures of sustainable development such as the Human Development Index (HDI)~\cite{neumayer2001human}, and Happy Planet Index~\cite{bondarchik2016improving} have life expectancy or longevity as their primary component. However, such an interpretation poses challenges in designing interventions. The expected lifespan of a geographical region may extend to several centuries, making it infeasible to verify any assertion about sustainability over time. Also, time flows in a single direction, making it infeasible to reverse a prior intervention.
Sustainability can also be defined alternatively as the ability of an entity to prevail against perturbations i.e any unexpected or adversarial change that  disturbs its current state~\cite{adger2000social,rose2014economic}.
Thus, in the context of intervention design we interpret sustainability as – \textbf{resilience against perturbations}. For instance, the sudden failure of a transmission line in power grid   or an ecosystem disrupted by an invasive species are both classic perturbations. The resilience of the power grid or the lake ecosystem depends on its ability to maintain functionality, despite the perturbation. 

Similarly, in agent-based social networks, a perturbation might involve the loss of a key influencers, or a shift in the incentives of its agents. From an evaluation standpoint, a complex system's state is said to be resilient if it can withstand targeted attacks and random perturbations against it~\cite{ResilienceDef,TAdef}. Hence, the objective of a sustainable intervention is to ensure that agents continue to make desired choices even in the case of perturbations i.e  ensure resilience of the system while achieving desirable state. 

\subsection{Building Resilience}
The  methods to achieve resilience in complex systems has been studied extensively. However, the the methods vary drastically from technical to economic to social domains.

Public policies are typically large-scale interventions designed to steer complex systems, such as nations or states, toward desired outcomes. Resilience frameworks for public policy recommend processes to embed resilience within these systems. For instance, the RPD Framework~\cite{ResilienceInPolicy2Framework} advocates combining policies that have historically performed well across dimensions like robustness, adaptability, and transformability to enhance resilience in bio-based production systems (BBPS). Similarly, a separate study identifies seven principles for designing resilient policies in social-ecological systems~\cite{ResilienceInPolicy1}, emphasizing policy diversity, goal redundancy reduction, and poly-centric governance. Implementing these recommendations requires evaluators to possess interdisciplinary expertise and contextual understanding of the target demographic and governance structures. Moreover, such frameworks are not directly transferable to other domains, such as telecommunications.

In telecommunications, resilience is achieved through strategies like backup topologies and dynamic routing~\cite{ResilienceTelecommunication}. These bear no conceptual commonality with  processes recommended to foster psychological resilience in post-trauma patients~\cite{ResiliencePostTrauma}, or to organizational resilience practices in Agile, Lean, and Six Sigma frameworks~\cite{AgileResilience,LeanSigSigmaResilience}. While certain resilience-building principles may appear across domains, they remain largely qualitative and context-specific, limiting direct cross-domain applicability.

Extensive research has focused on building resilience into networks, which can be applied to systems represented as networks~\cite{NetworkResilience2,NetworkResilience}. These studies typically define resilience through structural measures such as degree centrality and network connectivity. While they offer quantitative frameworks to enhance resilience, applying such methods to real-world social systems is challenging, as altering the topology of social networks is not trivial. Moreover, these approaches often assume resilience itself as the primary objective of intervention. In practice, however, interventions aim primarily to nudge populations toward desirable behaviours. Resilience is a desirable outcome that must be achieved alongside the primary aim.

There exists a fundamental postulate of sustainability, inspired by nature, proposed in the book Theory of Being~\cite{TheoryofBeing} as:\\
\textit{``Just about every system that is reasonably separable from its environment, seems to  eventually settle down into \textbf{low energy} stable regions.''}\\
Lower energy states are more sustainable and this is commonly observed to be the driving force behind various natural phenomena - the formation of molecular structures, chemical reactions, electron emissions, the movement towards hydrostatic equilibrium. This postulate is the base for optimization techniques such as simulated and quantum annealing~\cite{simulated_annealing}. In this work we explore leveraging this energy-based postulate to build resilience into interventions that are designed to bring behavioural change.

\subsection{Energy of a Complex System}

In complex systems, "energy" can be defined contextually, often reflecting properties of the system’s entities. In physical systems, energy manifests in mechanical, chemical, thermal, electrical, or nuclear forms. In social systems, several studies have tried to model energy-like potential functions that guide system dynamics~\cite{SocialEnergy,SocialEnergy2}. In behavioural contexts, agents’ alignment with their internal motivators serves as an analogue to energy. Similar to physical potential energy, agents exhibit higher energy when forced into states contradictory to their motivators i.e far from their preferred configurations. The Fogg Behaviour Model~\cite{FoggBehaviourModel} identifies three primary motivators for behaviour:
\begin{enumerate}
\item Pleasure/Pain,
\item Hope/Fear,
\item Social Acceptance/Rejection (Conformity).
\end{enumerate}

The Asch conformity experiments~\cite{solomon1956conformity} demonstrate that individuals often defer to the majority, especially in ambiguous situations, assuming others possess better knowledge. Conformity is thus a critical motivator in behavioural interventions. In this work, we model agents who are motivated by conformity. Their choices are influenced by their neighbours, motivated by a desire to align with the majority. An agent’s dissonance with its neighbours generates \textit{stress}, which we model as the agent’s energy, consistent with studies on cognitive dissonance~\cite{Conformity_Stress,CognitiveDissconance1,CognitiveDissconance2}. The system’s overall energy is thus captured by aggregate stress.

It is important to note that, other behaviour motivators may dominate in different contexts. For example, in interventions aimed at promoting vaccination, fear of side effects can serve as a source of stress. A population characterized by high levels of fear represents a high-energy state. 

\subsection{Energy-Based Intervention Design for Resilience}

Having defined energy within behavioural contexts, this work investigates the following hypotheses:

\textbf{H1:} The lower the energy of a complex social system, the higher its resilience. For conformity-driven agents, this translates to exploring the relationship between system resilience and the aggregate stress experienced by agents due to dissonance between their beliefs and those of their neighbours.

\textbf{H2:} Resilience can be embedded into interventions by incorporating energy-based constraints into their design. This allows resilience to be achieved alongside the primary objective of nudging agents toward desired behaviours.

The following sections, we presents a simulation study to investigate these hypotheses. Using an illustrative example based on commuter choices, we demonstrate how the energy-based approach can be integrated into intervention design to promote both effective and resilient behavioural change.

\section{Nudging Commuter Behaviour}\label{sec:network_setup} 

\subsection{Commuter Choice as a Stag Hunt Game}

Consider scenarios where individuals must choose between a socially responsible action and a more convenient, self-serving alternative—such as opting for cycling versus driving a car. The collective benefit is maximised when all individuals select the responsible option, as widespread adoption of cycling leads to improved environmental outcomes. However, in the absence of trust regarding others' cooperation, individuals may prefer the selfish alternative. In such situations, individuals who initially choose the responsible action may experience negative externalities; for example, cyclists may face traffic congestion and pollution generated by cars. Consequently, the  rational choice becomes defection — opting to drive car, especially when trust on widespread cooperation is low. This situation can be modelled using the Stag Hunt game~\cite{Stag_hunt_basics,StagHunt} (Table~\ref{tab:stag-hunt}).

\begin{table}[]
\centering
\caption{The payoff matrix for the stag-hunt game. Captures the essence of the decision between a responsible choice and a selfish choice.}
\begin{tabular}{|l|l|l|}
\hline
          & Cooperate                                                   & Defect                                                      \\ \hline
Cooperate & {\textbf{3,3}} & 0,2                                                         \\ \hline
Defect    & 2,0                                                         & {\textbf{2,2}} \\ \hline
\end{tabular}\label{tab:stag-hunt}
\end{table}

The stag hunt matrix has two pure Nash equilibria, 
\begin{enumerate}
    \item Universal Cooperation (Cooperate, Cooperate) - where everyone chooses to trust their neighbours and choose to cycle.
    \item Universal Defection (Defect, Defect) - where everyone chooses to drive car as they do not trust others to cooperate, and want to avoid sucker's payoffs.
\end{enumerate}

Let us assume policy makers want to nudge the population to be environment-friendly. They must design an intervention to push the population towards Universal Cooperation i.e when everyone choose to cycle. They can simulate the population as system of agents. As discussed before, conformity is a major motivator of agent behaviour when agents are placed in ambiguous situations. Thus agents are motivated to align to beliefs of their physical neighbours, as well as agents they are connected to via social networks. To incorporate this, every agent in the simulation is part of two networks -- one is based on geographical locations, and another is based on social interactions.

Each agent is assigned a geographical position and is identified by its \textit{latitude} and \textit{longitude}. As shown in Figure~\ref{fig: G1_Inital}, agents are arranged in neighbourhood clusters. Each agent plays the stag-hunt game (see Table ~\ref{tab:stag-hunt}) with all the agents within a certain \textit{influence radius} (hyper-parameter) of itself. Thus the physical network is formed by adding edges between all agents that lie within a fixed \textit{radius} of each other. Even within the radius an agent choosing to drive affects agents cycling nearby significantly more than those far away. To incorporate this  diminishing influence with increase in the physical distance we model the payoff as function of distance:
\begin{equation}\label{attenuation}
Payoff_i^{final}=Payoff_i\times\Gamma^{d}
\end{equation}
where  $\Gamma\in[0,1]$ is a hyperparamater to attenuate payoff and $d$ is the Euclidean distance between the two agents.
\begin{equation*}\label{L2_d}
    d=\sqrt{(latitude_i-latitude_j)^2+ (longitude_i-longitude_j)^2}
\end{equation*}

The social network is modelled as a Barabási-Albert network (Figure~\ref{fig: G1_Inital}), which follows a preferential attachment mechanism where new nodes are more likely to connect to existing nodes with higher degrees. This model was selected because many real-world networks—including social media platforms, the World Wide Web, and various communication and collaboration networks—exhibit scale-free properties that are well captured by the Barabási-Albert structure~\cite{Barabasi1999}. In such networks, a small number of highly connected hubs coexist with a large number of nodes with few connections, reflecting the heterogeneity commonly observed in social systems.

\subsection{Strategies and Belief Revision}
At the beginning of the simulation, each agent \textbf{i} is randomly assigned a strategy $s_i \in [0,1]$, indicating the probability with which the agent \textbf{i} chooses to cycle. The agents can now be profiled based on their strategies. 
\begin{table}[]
\centering
\caption{Agents are profiled on the basis of their strategy ranges. $s_i$ or strategy of agent i indicates the probability of it choosing cycle}
\begin{tabular}{|l|l|}
\hline
Profile            & $s_i$ (Probability of Choosing Cycle) \\ \hline
Highly Distrusting & 0  to 0.25                    \\ \hline
Distrusting        & 0.25 to 0.50                  \\ \hline
Trusting           & 0.50 to 0.75                  \\ \hline
Highly Trusting    & 0.75 to 1                     \\ \hline
\end{tabular}
\end{table}
\begin{figure*}
\centering
  \begin{minipage}{0.45\textwidth}
     \includegraphics[width=\textwidth, height=\textwidth]{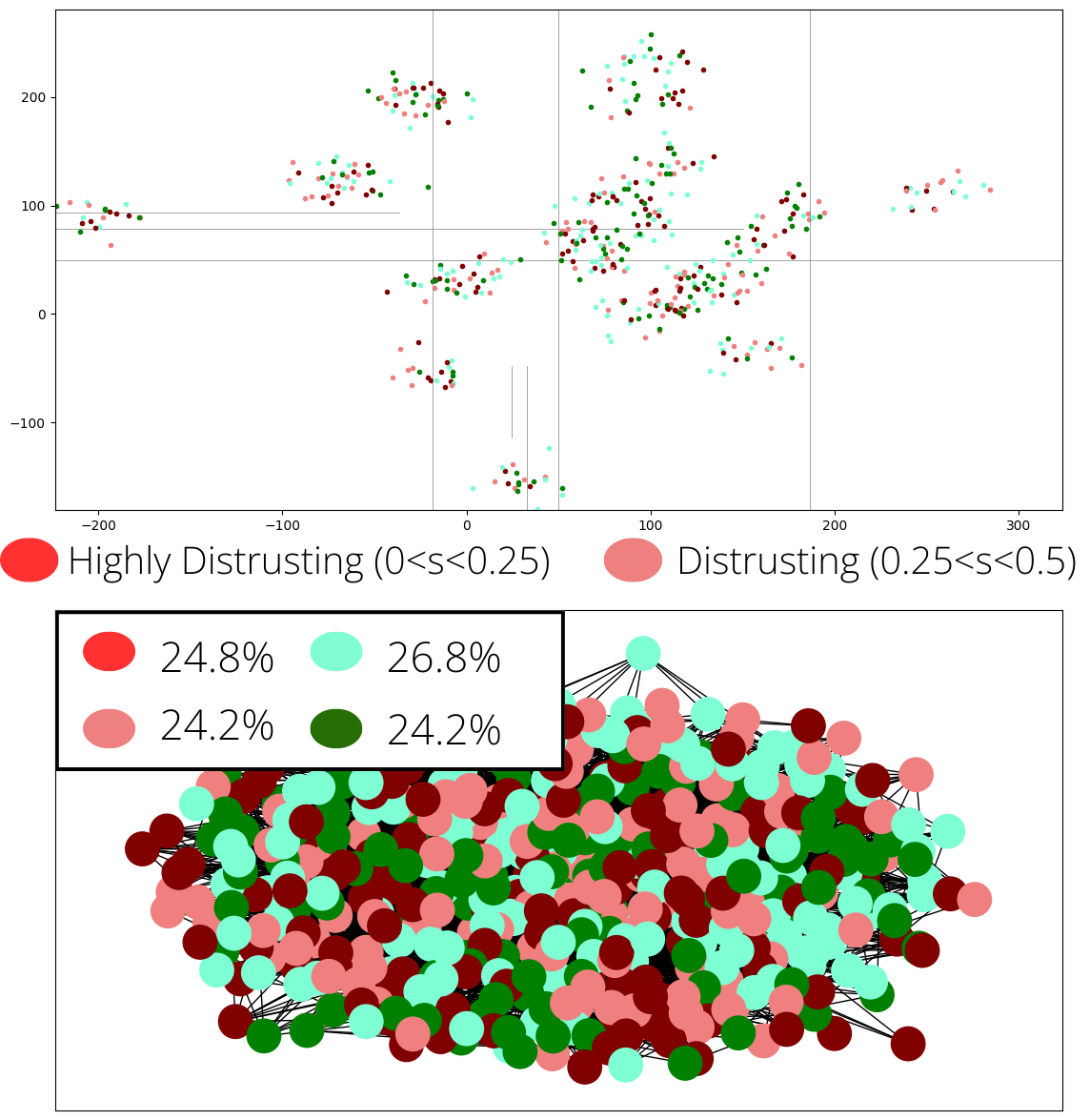}
  \end{minipage}
  \hfill
  \begin{minipage}{0.45\textwidth}\includegraphics[width=\textwidth,height=\textwidth]{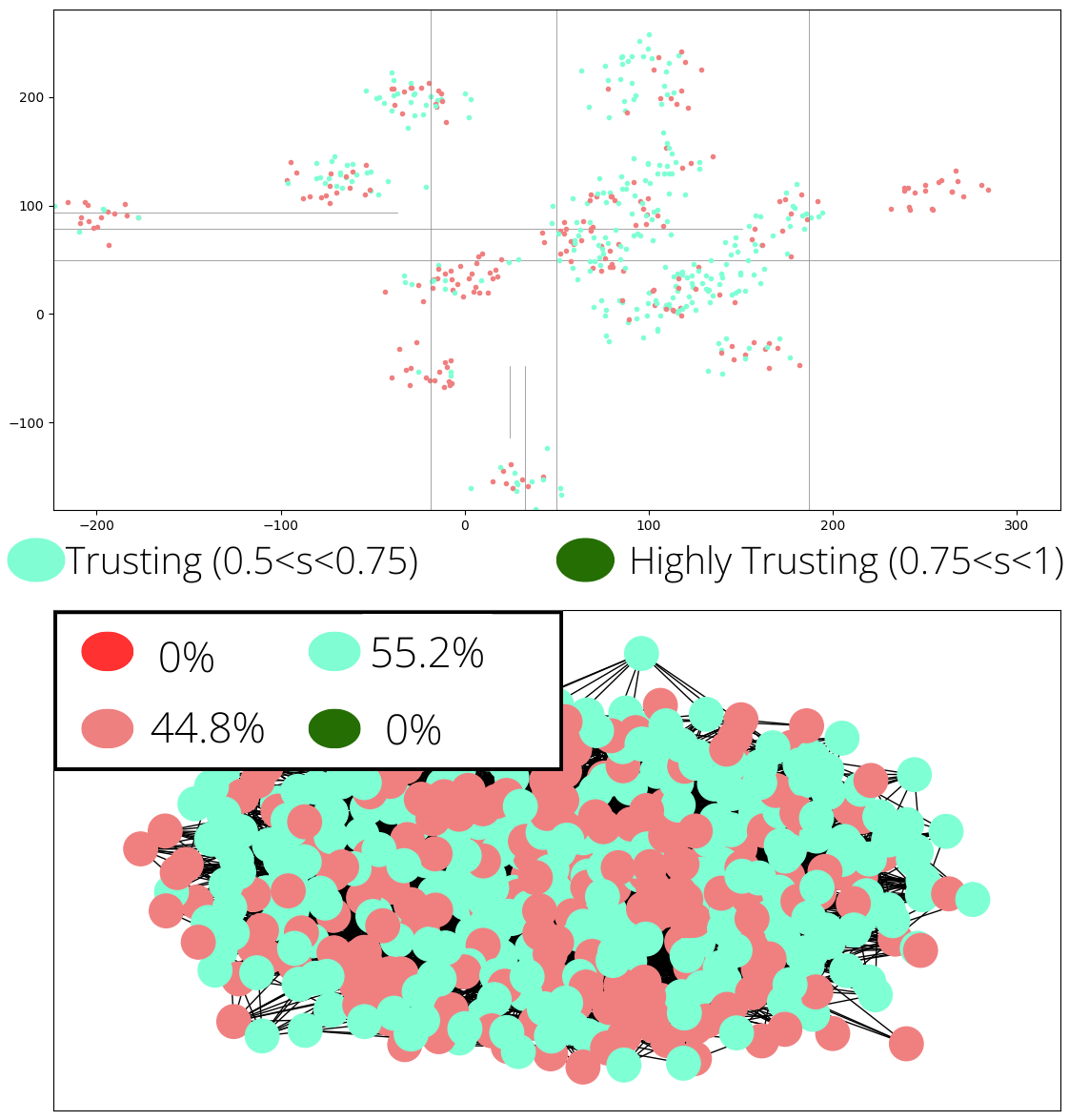}
  \end{minipage}
  \caption{\label{fig: G1_Inital} Conforming agents are motivated align to beliefs of their physical neighbours, as well as agents they are connected to via social networks. To capture this, every agent is part of two networks in the simulation. The graphs on top show the Physical Network, based on the position of the agent defined by a (longitude,latitude) coordinate. The bottom ones show the Social Network modelled as a Barabasi-Albert network. The agents are profiled based on their strategies (see Section~\ref{sec:network_setup}). Strategies are randomly assigned during initialization resulting in an almost equal distribution as seen on the left. The simulation converges to an initial state with the distribution shown on the right}
\end{figure*}

The game is played over several epochs; where each epoch consists of every agent playing the stag-hunt game with all other agents within the fixed \textit{radius}. After every epoch, agents undergo a belief revision either motivated by \textbf{conformity} or to \textbf{maximise utility}:
\subsubsection{Conformity Motivated Belief Revision}
In a conforming population, the dissonance between an agent's belief and that of its neighbours is captured as \textit{stress}. Stress experienced by agent \textbf{i} is:
\begin{multline}\label{stress}
stress_i = \sum_{j\in Physical\_neighbours(i)}L_2(s_j - s_i)+\sum_{k\in Social\_neighbours(i)}L_2(s_k - s_i)  
\end{multline}
The physical neighbourhood of an agent consists of all the agents whose physical location falls in the \textit{influence radius} of itself. The social neighbourhood of an agent consists of all agents who are directly connected to the agent in the social network layer.

The value of $s_i$ that minimizes the stress function above is the centroid of the polygon formed by the strategy values of all the neighbours of $i$ - physical and social\footnote{Refer to the Appendix for Proof: centroid minimizes equation~\ref{stress}.}. Thus the optimal strategy for an agent that is motivated by conformity and aims to experience the least stress from its neighbours is given by:
\begin{equation}\label{optimal}
    s^{opt}_i = \frac{\sum_{j\in Physical\_neighbours(i)}s_j +\sum_{k\in Social\_neighbours(i)}s_k}{|Physical\_neighbours|+|Social\_neighbours|}
\end{equation}
An agent motivated by conformity chooses to revise its beliefs from epoch \textit{t} to \textit{t+1}, following a gradient descent towards the optimal strategy with a learning rate $\Delta$ can be described as follows:
\begin{align*}
    s^{t+1}_i &= s^t_i + \Delta \times ( s^{opt}_i - s^{t}_i)\\
\end{align*}
\subsubsection{Utility Maximisation Motivated}
Agents could also \emph{rationally} choose to change their strategy to mimic the strategy of the best-performing agent within their physical neighbourhood, with a learning rate $\Delta$ can be described as: 

\begin{align*}
    k &= \arg\max_{j \in \text{Physical\_neighbours}(i)} \left( total\_Payoff_j \right) \\
    s^{t+1}_i &= s^t_i + \Delta \times (s^t_k - s^t_i)
\end{align*}

The hyper-parameter $\alpha \in [0,1]$ indicates the probability with which the belief revision is motivated by the agent's wish to conform with its neighbours. $ 1- \alpha$ is the probability of the belief revision being motivated by utility maximisation. While modelling a conforming population, we ensure $\alpha > 0.5$. This signifies that agents change their strategy mainly to align with the neighbours and behave purely rationally only a small fraction of the time.
\subsection{Metrics for Intervention Design}
Once the system is set up with $n$ agents, the simulation is run till convergence (or) for 1500 epochs, whichever happens first. If none of the agents revises their beliefs for 20 epochs, the simulation is considered to have converged. Post convergence state of the system can be evaluated based on the following metrics:
\subsubsection{Impact}
This is a metric used to evaluate an intervention. It compares the number of agents with cycle strategies ($s>0.5$) before and after the intervention and captures the increase.
\begin{equation}\label{impact}
impact\_score(intervention) = \frac{\sum_{i\in agents}s'_i>0.5 - \sum_{i\in agents}s_i>0.5}{\sum_{i\in agents}s_i>0.5}\\
\end{equation}
where $s$ is the agent strategy before the intervention and $s'$ is after.
\subsubsection{Stability}
The stability of a system state is determined by how close the agent strategies are to their optimal strategies.
\begin{equation}\label{stability}
    stability (state) = 1 - \frac{\sum_{i\in agents } L_2(s^{opt}_i-s_i)}{n}
\end{equation}
From Equation~\ref{optimal} we know that an agent experiences the least stress when its strategy is $s^{opt}_i$. At this point $L_2(s^{opt}_i-s_i)=0$. However, if an agent's strategy is very far from optimal strategy i.e  $L_2(s^{opt}_i-s_i)$ is high, then the agent becomes unstable. When the majority of the agents are not minimizing stress and are unstable,  then by equation~\ref{stability}, the state's stability decreases. The higher the stress on individual agents, the higher the energy and thus, the lower the stability of the network. Hence the \textit{\textbf{stability of a state acts as an indicator of its energy}}.
\subsubsection{Resilience}\label{sec:resilience}
The resilience of a system state is evaluated based on its ability to withstand certain targeted attacks and random perturbations~\cite{ResilienceDef,TAdef}. 
A random perturbation or a targeted attack may induce changes in the strategies of the agents. The smaller the magnitude of the change, the greater the resilience of the state. Thus the resilience score of a state to a particular attack/perturbation can be defined as:
\begin{equation}\label{resilience}
resilience\_score(state) = \frac{1}{e^{\delta(S,S')}}\\
\end{equation}
\begin{align*}
S &= \{s_1,s_2,...,s_{n}\}  \\
S' &= \{s'_1,s'_2,...,s'_{n}\}  \\
\delta(S,S') &= \sqrt{\sum_{i\in agents}(s_i-s'_i)^2}
\end{align*}
where $S$ is the set of agent strategies before the attack/perturbation and $S'$ is after. Since they can be considered as vectors, the Euclidean distance between $S$ and $S'$ is taken as a measure of the magnitude of the attack/target. 

\section{Experiments and Results}\label{sec:experiments_results}  
The system configurations in this section consist of 500 agents, and follow the setup discussed in Section~\ref{sec:network_setup}. On these simulations are run with the following hyperparameter values -- $
    \alpha = 0.7,  
    \gamma = 0.8,  
    \Delta = 0.01
$

By setting $\alpha=0.7$, the agents are motivated to change their strategies based on their desire to conform (70\% of the time) and to maximise utility (30\% of the time), forming a largely conforming population. This particular combination of hyperparameters was chosen as they caused the initial state to reach an almost equal share of car and cycle users. This gives a proper environment to introduce interventions to push the population towards cycle users and evaluate the post-intervention states on impact and resilience. However, the findings of this work hold for all populations where $\alpha>0.5$.

\subsection{Hypothesis 1: Lower the energy of a complex social system, greater the resilience}\label{sec:stabilityResilience}
Twenty-five random system configurations were generated to investigate the correlation between resilience and stability, and consequently, between resilience and energy of the system. A minimum sample size of 25 observations is recommended for reliable Pearson correlation calculations ~\cite{DavidSmallSamplesize}. Since we employ the coefficient of determination ($r^2$) in our analysis, which is derived from Pearson’s correlation coefficient, we adopted this recommendation to ensure the robustness of our correlation estimates.
\begin{table*}[]
    \centering
    \caption{Systems' resilience scores to a particular attack/perturbation plotted against their stability values. The R-squared values indicate the strength of the correlation between resilience and stability, and thus resilience and energy of the system.}
    \begin{tabular}{|c|c|}
        \hline
        Attack on hubs & Attack by connecting car-users\\
       \includegraphics[width=0.43\textwidth]{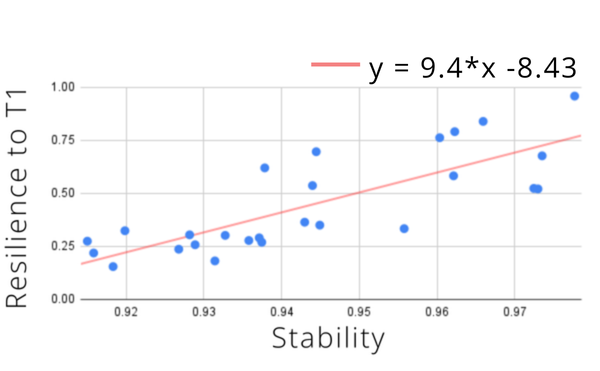}  &  
       \includegraphics[width=0.43\textwidth]{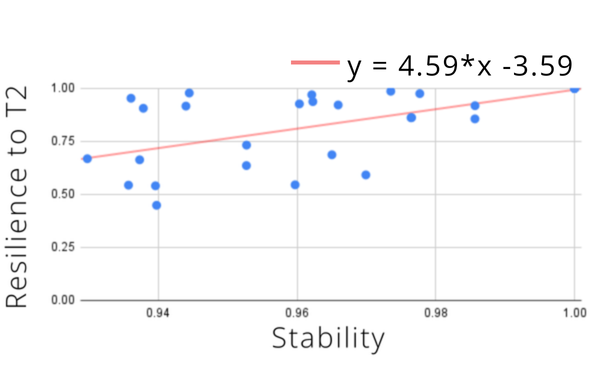} \\
       \footnotesize{R-squared: 0.633} & \footnotesize{R-squared: 0.453} \\
        \hline
       Random loss of social connections & Random agent drop-out\\
       \includegraphics[width=0.43\textwidth]{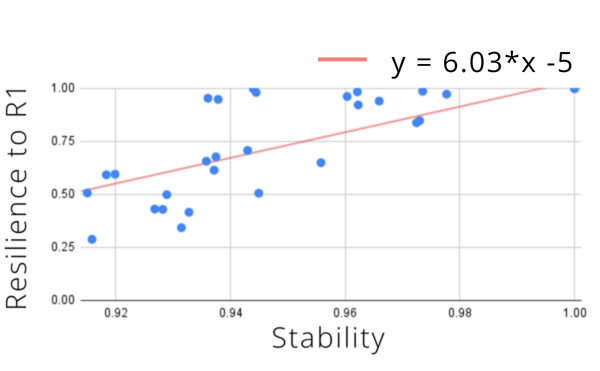} & 
       \includegraphics[width=0.43\textwidth]{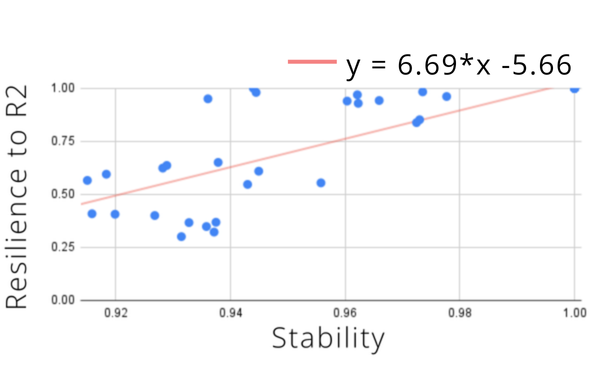}\\
       \footnotesize{R-squared: 0.645} & \footnotesize{R-squared: 0.65} \\  
        \hline
    \end{tabular}
    \label{fig:Resilience_stability}
\end{table*}
Agents in the system were initially assigned strategies using a uniform random function. As a result, the four strategy profile categories were evenly distributed at the start (see Figure~\ref{fig: G1_Inital}). Then the configurations were allowed to converge and the stabilities of these initial states were noted. Then, the initial states were subjected to the following targeted attacks and random perturbations\footnote{Refer to the Appendix for the detailed algorithms of the attacks and perturbation simulated in this work.}:
\begin{enumerate}
    \item \textit{Targeted Attack 1 - Forced Strategy Change of Hubs}\label{attack: T1}: In this attack, the strategy of agents which are identified as hubs of the social network is forcefully changed of cars (i.e. towards defecting). In a conforming population, this then encourages the other influenced agents to follow suit.
    \item \textit{Targeted Attack 2 - Car Users Come Together}\label{attack: T2}: In this attack, new social connections are proposed between car users(defectors) within a two-hop neighbourhood, which then get accepted with some probability. Since the agents are motivated to conform, the strategy of car users who accept the new social connections gets highly reinforced due to the addition of similar neighbours.
    \item \textit{Random Perturbation 1 - Broken Connections in the Social Network}\label{attack: R1}: In this perturbation, some random edges in the social network are removed. This translates to a loss of social connections between two humans, which is a common occurrence.
    \item \textit{Random Perturbation 2 - Agents Drop-out of the Social Network}\label{attack: R2}: In this perturbation, some random agents are removed from the social network i.e the in and out edges of the agent are removed from the social network. This translates to the end  of an individual's influence in a community, which is also quite common.
\end{enumerate}
 The resilience score to each attack and perturbation was calculated for all the 25 systems. Figure~\ref{fig:Resilience_stability} shows the resulting resilience vs stability graphs. The best-fit line for each scatter plot was obtained using the least square method. The best-fit lines in all four cases have positive slopes . A positive slope implies a correlation; hence, we can surmise that the greater the stability of a state of the system, greater its resilience to most attacks and perturbations. The $r^2$ values of the best-fit lines for resilience score for the targeted attack on hubs and both the random attacks are nearly 0.65. While this does not indicate a perfect linear relationship, it suggests that approximately 65\% of the variance in resilience can be explained by the system’s energy, supporting the hypothesis that \emph{higher the stability and lower the energy of a system, the higher its resilience}.

\subsection{Hypothesis 2: Embedding Energy Constraints into Intervention Design builds Resilience}
\begin{table}[]
\caption{\label{tab:g2_Intervention_compare} Comparing the vanilla CPI to its variant which prioritizes stability. Here a Targeted Attack against hubs was simulated (T1)}
\begin{tabular}{|l|l|l|l|l|l|}
\hline
 &
  \begin{tabular}[c]{@{}l@{}}Cyclists Before\\ Intervention\end{tabular} &
  \begin{tabular}[c]{@{}l@{}}Cyclists After\\ Intervention\end{tabular} &
  \begin{tabular}[c]{@{}l@{}}Impact\\ Score\end{tabular} &
  \begin{tabular}[c]{@{}l@{}}Cyclists\\ After Attack\end{tabular} &
  \begin{tabular}[c]{@{}l@{}}Resilience\\ Score\end{tabular} \\ \hline
\begin{tabular}[c]{@{}l@{}}Intervention 1\\ (Vanilla CPI)\end{tabular} &
  55.2\% &
  72\% &
  0.3 &
  12.2\% &
  0.42 \\ \hline
\begin{tabular}[c]{@{}l@{}}Intervention 2\\ (Stable CPI)\end{tabular} &
  55.2\% &
  68.2\% &
  0.23 &
  \textbf{55.4}\% &
  \textbf{0.82} \\ \hline
\end{tabular}
\end{table}

A system configuration is generated and initial strategies were assigned to the 500 agents based on a uniformly random function once again, resulting in equal distribution among the four strategy profiles. Then it was run till convergence to get to an initial state which had 44.8\%  agents inclined towards driving car and 55.2\% agents inclined towards cycling. Since there are still a significant number of car agents, the state is sub-optimal. To increase the number of cycle users, we performed an intervention.
\begin{table*}
\centering
\caption{Comparing resilience of the system to Random Perturbations 1 and 2 post Vanilla CPI (I1) and Stable CPI(I2)}
\begin{tabular}{|l|rr|rr|}
\hline
\begin{tabular}[c]{@{}l@{}}
Simulation\\  No.\end{tabular}& \multicolumn{2}{l|}{Resilience to R1}                      & \multicolumn{2}{l|}{Resilience to R2}                      \\ \hline
                                                               & \multicolumn{1}{l|}{I1}    & \multicolumn{1}{l|}{I2}       & \multicolumn{1}{l|}{I1}    & \multicolumn{1}{l|}{I2}       \\ \hline
1                                                         & \multicolumn{1}{r|}{0.59}  & 0.60  & \multicolumn{1}{r|}{0.54}  & 0.58  \\ \hline
2                                                         & \multicolumn{1}{r|}{0.75}  & 0.88  & \multicolumn{1}{r|}{0.91} & 0.92 \\ \hline
3                                                         & \multicolumn{1}{r|}{0.60}  & 0.60                          & \multicolumn{1}{r|}{0.64}  & 0.67  \\ \hline
4                                                          & \multicolumn{1}{r|}{0.86} & 0.86 & \multicolumn{1}{r|}{0.86}  & 0.86 \\ \hline
5                                                       & \multicolumn{1}{r|}{0.73}  & 0.74 & \multicolumn{1}{r|}{0.73} &0.74 \\ \hline
Average Score                                             & \multicolumn{1}{r|}{0.73}  & 0.77  & \multicolumn{1}{r|}{0.78}  & 0.79  \\ \hline
\end{tabular}

\label{tab:resilience_Simulations}
\end{table*}
\subsubsection{Intervention 1: Vanilla CPI}
\footnote{Refer to the Appendix for details on Intervention Implementations.}
In a largely conforming population, one common way to reinforce people's beliefs and prevent unwanted strategy changes is to connect them to others who uphold the same beliefs . This intervention is known as the Connect People Intervention(CPI)~\cite{ConnectPeopleIntervention}. Even in social media networks, there is a feature that recommends people from our second or third circles based on similarity of interests and beliefs. The CPI is simulated by proposing connections between cycle users who are within two hops of each other. Assuming that agents accept the suggested connections with an \textbf{acceptance\_probability} of 0.15 ; 15\% of the suggested social connections are added to the network. The simulation is run again till convergence. The impact of the vanilla CPI  is measured. Targeted Attack 1 (see Section ~\ref{sec:stabilityResilience}) is performed by forcefully setting the strategies of the top 30 agents, ordered by \textit{hub score}, to \textbf{0.3} i.e. 70\% inclined to cars. The resilience of the system to this Targeted Attack is measured. 
\subsubsection{Intervention 2: Stable CPI}
Simulation of this Intervention also involves recommending cycle users to other cycle users within their two-hop neighbourhood. However, similarity is not the only criterion. The intervention design is tweaked and a new constraint is added. Connections are only proposed if the addition of the edge will reduce the stress of the agents and thereby, improve the overall stability of the network. Again, we assume that agents accept the suggested connections with \textbf{acceptance\_probability} of 0.15. The simulation is run  till convergence and the impact of the stable version of CPI and resilience of the post-intervention to Targeted Attack on hubs is measured.
\subsubsection{Comparing Interventions 1 and 2}
Intervention 2, whose design was tweaked to maintain stability, fared better on resilience than Intervention 1, while creating a comparable impact (Table ~\ref{tab:g2_Intervention_compare}). This can also be observed clearly in Table~\ref{tab:resilience_Simulations} which compares the resilience of the post-intervention states to perturbations, after being subject to vanilla CPI and stable CPI. Stable CPI fared better on the resilience metric. This is in line with the findings from the previous section that stability leads to resilience. Thus, embedding energy minimization into intervention design builds resilience.

\section{Discussion}

The results of investigating hypothesis 1 showed that while the relationship between resilience and stability may not be strictly linear, however it does provide evidence that higher resilience is generally observed in social systems with lower energy, consistent with the fundamental energy postulate. Consider a scenario where an evaluator has to pick the best intervention to implement from among  $m$ options with comparable impact. Ideally, it would require evaluation of \textbf{resilience} of the $m$ interventions. However, guided by the principle that low-energy states are more resilient, the evaluator can simply select the intervention that yields the lowest-energy state. This saves both time and computational effort otherwise required to simulate and analyze a wide range of perturbations on each of the $m$ interventions.

The investigations for hypothesis 2 demonstrated that energy-aware behavioural intervention design contributes to building resilience. An intervention designer must simultaneously optimize two objectives: (i) achieving the desired behavioural change, and (ii) ensuring resilience of the resulting state. In this study, the behavioural objective was to maximise the number of agents choosing to cycle. Resilience was promoted by incorporating an energy-based constraint: agents are not allowed to form connections if doing so increases their dissonance and thereby their energy. By integrating this regularization mechanism into the intervention design, the Stable CPI ensured that the system remains close to a low-energy configuration while achieving comparable impact to the vanilla version. Thus it created \textbf{sustainable behavioural change}.

Intuitively, the system’s energy reflects how well agents' behaviours align with their fundamental motivations. Incorporating stability-based regularization allows the intervention to respect these underlying motivations, avoiding scenarios where agents are pushed into behaviours that strongly contradict their natural preferred state. The larger the number of agents experiencing such misalignment from their desired state, the more susceptible the system becomes to perturbations, increasing the likelihood of collapsing to an undesired state. Thus the impact it creates is not sustainable.  
\begin{figure*}
\centering
  \begin{minipage}{0.6\textwidth}
     \includegraphics[width=\textwidth, height=0.6\textwidth]{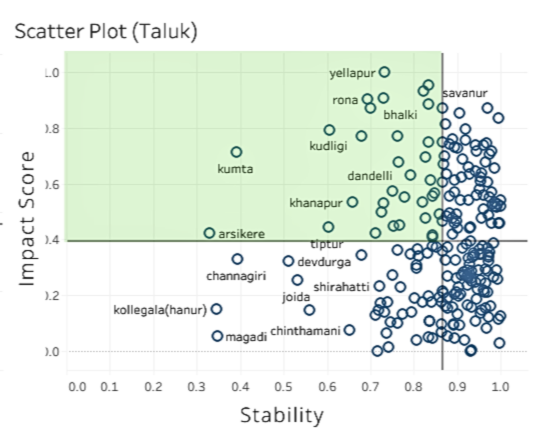}
  \end{minipage}
  \hfill
  \begin{minipage}{0.35\textwidth}\includegraphics[width=\textwidth,height=\textwidth]{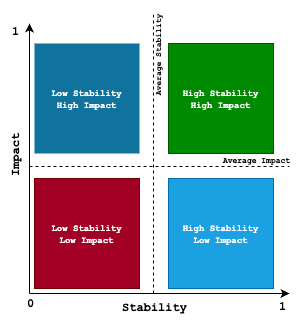}
  \end{minipage}
  \caption{\label{fig:Conclusion_image2}Impact vs Stability graph for Infant Mortality Rate. Sub-districts in the top right quadrant are the best candidates to implement the intervention}
\end{figure*}

The commuter's choice was a simple set-up used in the paper to investigate the hypotheses in a quantitative way within a simulation, with many inherent assumptions. The energy-based principle to building resilience can be incorporated into more intricate intervention designs for complex systems. Similar principles have been previously used in policy-making and intervention design for domains like healthcare, agriculture, climate change etc ~\cite{2023IJCAI}. For instance, energy-based intervention evaluation have been presented in the Karnataka Data Lake Website\footnote{\url{https://avalokana.karnataka.gov.in/DataLake/DataLake}}. Figure~\ref{fig:Conclusion_image2} represents an intervention aimed at improving one of the UN Sustainable Development Goal (UN-SDG)- Infant Mortality Rate (IMR). The intervention is simulated for all the sub-districts of the Indian state of Karnataka. The sub-districts where both impact and post-intervention stability are high are the best places to intervene as they will result in sustainable outcomes. Thus sub-districts in the top right quadrant are recommended to the policy makers as the best candidates to implement the intervention.

\section {Conclusion and Future Work}
This work offers a perspective towards building sustainability into intervention design using a energy lens.  It is a cross-disciplinary view integrating physics, ecology, sociology, and systems thinking. It offers the holistic view for policy-design showcasing resilience can be built into interventions designed for different contexts. This is achieved by only selecting the interventions that don't increase the energy of the social systems on which they are applied. 

While this work can be expanded to accurately design interventions for complex use-cases, it also gives a simple rule of thumb to evaluate any behavioural nudge. If it forces a very sudden, very large misalignment between agents behaviour and their motivations - fear, social rejection, pain - while trying to achieve its aim, the final state achieved will not be sustainable. 

While we chose the context of nudging the population towards cycles as an example, it would be interesting to explore how popular interventions aimed only at behavioural change, can be enhanced to achieve resilience by tweaking their design to be energy aware, thereby creating sustainable impact.

%
%
%
\bibliographystyle{splncs04}
\bibliography{references}
\begin{center}
\clearpage
\appendix
\Large\bfseries Appendix
\end{center}
\section{Targeted Attacks}\label{sec:TA}
\begin{algorithm}[H]
\caption{ Forced Strategy Change of Hubs}
\begin{algorithmic}
\Procedure{AttackHubs}{$attack\_magnitude$,$new\_startegy$}
    \State $Hubs \gets Sorted(agents, key \gets hub\_score)$
    \For{$i\gets 0, attack\_magnitude$}
        \State $s_{Hubs[i]} \gets new\_startegy$
    \EndFor
\EndProcedure
\end{algorithmic}
\end{algorithm}
\begin{algorithm}[H]
\caption{Car Users Come Together}
\begin{algorithmic}
\Procedure{ConnectDefectors}{$acceptance\_probability$}
    \ForAll{$ i \in |agents|$}
        \ForAll {$j \in social\_neighbourhood(i)$}
            \ForAll {$k \in social\_neighbourhood(j)$}
                \If{$s_i<0.5$ and $s_k<0.5$}
                    \If{$random()<acceptance\_probability$}
                        \State append $Edge(i,k)$ to $Edges(social\_network)$
                    \EndIf
                \EndIf
            \EndFor
        \EndFor
    \EndFor
\EndProcedure
\end{algorithmic}
\end{algorithm}
\clearpage
\section{Random Perturbations}\label{sec:RP}
\begin{algorithm}[H]
\caption{Broken Connections in the Social Network}
\begin{algorithmic}
\Procedure{ConnectionBreaks}{$perturbation\_magnitude$}
    \ForAll {$edge \in Edges(social\_network)$}
        \If{$random()<perturbation\_magnitude$}
            \State remove $edge$ from $Edges(social\_network)$
        \EndIf
    \EndFor
\EndProcedure
\end{algorithmic}
\end{algorithm}
\begin{algorithm}[H]
\caption{Agents Drop-out of the Social Network}
\begin{algorithmic}
\Procedure{AgentsDropOut}{$perturbation\_magnitude$}
    \ForAll{$ i \in |agents|$}
        \If{$random()<perturbation\_magnitude$}
            \State $\emph{E} \gets \{e \in Edges(social\_network)$ with i=source or i=destination $\}$
            \State remove $\emph{E}$ from $Edges(social\_network)$
        \EndIf  
    \EndFor
\EndProcedure
\end{algorithmic}
\end{algorithm}
\section{Interventions}\label{sec:Interventions}
\begin{algorithm}[H]
\caption{Vanilla CPI}
\begin{algorithmic}
\Procedure{Intervention1}{$acceptance\_probability$}
    \ForAll{$ i \in |agents|$}
        \ForAll {$j \in social\_neighbourhood(i)$}
            \ForAll {$k \in social\_neighbourhood(j)$}
                \If{$s_i>0.5$ and $s_j>0.5$ and $s_k>0.5$}
                    \If{$random()<acceptance\_probability$}
                        \State append $Edge(i,k)$ to $Edges(social\_network)$
                    \EndIf
                \EndIf
            \EndFor
        \EndFor
    \EndFor
\EndProcedure
\end{algorithmic}
\end{algorithm}
\begin{algorithm}[H]
\caption{Stable CPI}
\begin{algorithmic}
\Procedure{Intervention2}{$acceptance\_probability$}
    \ForAll{$ i \in |agents|$}
        \ForAll {$j \in social\_neighbourhood(i)$}
            \ForAll {$k \in social\_neighbourhood(j)$}
                \If{$s_i , s_j, s_k >0.5$ and $Edge(i,k) \notin Edges(social\_network)$}
                    \State$\Delta_{stability}\gets$
                    \begin{align*}
                        (L_2(s_i,s_i^{opt})+L_2(s_k,s_k^{opt}))-\\(L_2(s_i,s_i^{opt}(including Edge(i,k)))+L_2(s_k,s_k^{opt}(including Edge(i,k))))
                    \end{align*}
                    \If{$random()<acceptance\_probability$ and $\Delta_{stability}>0$}
                        \State append $Edge(i,k)$ to $Edges(social\_network)$
                    \EndIf
                \EndIf
            \EndFor
        \EndFor
    \EndFor
\EndProcedure
\end{algorithmic}
\end{algorithm}
\section{Theorems}\label{sec:proofs}
\begin{theorem}
    Centroid Minimizes the Sum of Squared Euclidean Distances
\end{theorem}
\begin{proof}
Let $(x_n, y_n)$ be $n$ points. We aim to minimize:
\begin{align*}
    f &= \sum_{i=1}^n (x - x_i)^2 + (y - y_i)^2 \\
    \frac{\partial f}{\partial x} &= \sum_{i=1}^n \frac{\partial}{\partial x} \left[ (x - x_i)^2 + (y - y_i)^2 \right] \\
    &= \sum_{i=1}^n 2(x - x_i)
\end{align*}
Setting the derivative to zero:
\begin{align*}
    \frac{\partial f}{\partial x} = 0 &\Rightarrow \sum_{i=1}^n 2(x_{\min} - x_i) = 0 \\
    &\Rightarrow \sum_{i=1}^n (x_{\min} - x_i) = 0 \\
    &\Rightarrow n \cdot x_{\min} = \sum_{i=1}^n x_i \\
    &\Rightarrow x_{\min} = \frac{1}{n} \sum_{i=1}^n x_i
\end{align*}
Similarly,
\[
    y_{\min} = \frac{1}{n} \sum_{i=1}^n y_i
\]
\end{proof}
\end{document}